\documentclass[11pt,a4paper]{article}
\usepackage{fullpage}
\usepackage{booktabs}
\usepackage{multirow}
\usepackage[utf8]{inputenc}
\usepackage{amsthm,amssymb}
\usepackage[leqno]{amsmath}
\usepackage{tcolorbox}
\usepackage{thmtools}
\usepackage{mathtools}
\usepackage{subcaption,thm-restate}
\usepackage[mathcal,mathscr]{eucal}
\usepackage{epsfig,graphicx,graphics,color}
\usepackage{caption}
\usepackage{enumerate}
\usepackage{enumitem}
\usepackage[sort]{cite}
\usepackage{titling}
\usepackage{hyperref}
\hypersetup{
    colorlinks=true,       
    linkcolor=blue,        
    citecolor=red,         
    filecolor=magenta,     
    urlcolor=cyan,         
    linktocpage=true
}
\usepackage{thm-restate}
\usepackage{cleveref}
\usepackage{float}
\usepackage{algpseudocode}
\usepackage{rotating}
\usepackage{makecell}
\usepackage[ruled,noline, linesnumbered]{algorithm2e}

\usepackage{xcolor}
\definecolor[named]{ACMBlue}{cmyk}{1,0.1,0,0.1}
\definecolor[named]{ACMYellow}{cmyk}{0,0.16,1,0}
\definecolor[named]{ACMOrange}{cmyk}{0,0.42,1,0.01}
\definecolor[named]{ACMRed}{cmyk}{0,0.90,0.86,0}
\definecolor[named]{ACMLightBlue}{cmyk}{0.49,0.01,0,0}
\definecolor[named]{ACMGreen}{cmyk}{0.20,0,1,0.19}
\definecolor[named]{ACMPurple}{cmyk}{0.55,1,0,0.15}
\definecolor[named]{ACMDarkBlue}{cmyk}{1,0.58,0,0.21}
\definecolor{green}{rgb}{0,0.6,0}
\definecolor{fillblack}{rgb}{0.95,0.95,0.2195}

\theoremstyle{plain}
\newtheorem{theorem}{Theorem}
\newtheorem{lemma}[theorem]{Lemma}

\newtheorem{claim}[theorem]{Claim}

\theoremstyle{definition}

\newtheorem{rem}[theorem]{Remark}

\newcounter{casenum}

\DeclareMathOperator*{\argmax}{arg\,max}

\def\final{0}  
\def\iflong{\iffalse}
\ifnum\final=0  
\newcommand{\anote}[1]{{\color{orange}[{\tiny \textbf{Alex:} \bf #1}]\marginpar{\color{orange}*}}}
\newcommand{\knote}[1]{{\color{red}[{\tiny \textbf{Krist{\'o}f:} \bf #1}]\marginpar{\color{red}*}}}
\newcommand{\jnote}[1]{{\color{blue}[{\tiny \textbf{Julian:} \bf #1}]\marginpar{\color{blue}*}}}
\newcommand{\lnote}[1]{{\color{purple}[{\tiny \textbf{Laura:} \bf #1}]\marginpar{\color{purple}*}}}
\else 
\newcommand{\anote}[1]{}
\newcommand{\knote}[1]{}
\newcommand{\jnote}[1]{}
\newcommand{\lnote}[1]{}
\fi

\usepackage{authblk}

\title{Envy-free dynamic pricing schemes}

\date{}

\thanksmarkseries{alph}
\author{Krist{\'o}f B{\'e}rczi} 
\affil{{\footnotesize MTA-ELTE Momentum Matroid Optimization Research Group and MTA-ELTE Egerv\'ary Research Group, Department of Operations Research, E\"otv\"os Lor\'and University, Budapest, Hungary. Email: \texttt{kristof.berczi@ttk.elte.hu}.}}
\author{Laura Codazzi}
\author{Julian Golak}
\affil{{\footnotesize Institute of Algorithms and Complexity, Hamburg University of Technology, Hamburg, Germany. Email: \texttt{laura.codazzi@tuhh.de, julian.golak@tuhh.de}.}}
\author{Alexander Grigoriev}
\affil{{\footnotesize Department of Data Analytics and Digitalisation, Maastricht University, Maastricht, The Netherlands. Email: \texttt{a.grigoriev@maastrichtuniversity.nl}.}}

\newcommand{\setPl}{\mathcal{A}} 
\newcommand{\nPl}{n} 
\newcommand{\setIt}{\mathcal{I}} 
\newcommand{\nIt}{m} 
\newcommand{\player}{a} 

\begin{document}
\maketitle

\begin{abstract}

A combinatorial market consists of a set of indivisible items and a set of agents, where each agent has a valuation function that specifies for each subset of items its value for the given agent. From an optimization point of view, the goal is usually to determine a pair of pricing and allocation of the items that provides an efficient distribution of the resources, i.e., maximizes the social welfare, or is as profitable as possible for the seller, i.e., maximizes the revenue. To overcome the weaknesses of mechanisms operating with static prices, a recent line of research has concentrated on dynamic pricing schemes. In this model, agents arrive in an unspecified sequential order, and the prices can be updated between two agent-arrivals. Though the dynamic setting is capable of maximizing social welfare in various scenarios, the assumption that the agents arrive one after the other eliminates the standard concept of fairness. 

In this paper, we study the existence of optimal dynamic prices under fairness constraints in unit-demand markets. We propose four possible notions of envy-freeness of different strength depending on the time period over which agents compare themselves to others: the entire time horizon, only the past, only the future, or only the present. For social welfare maximization, while the first definition leads to Walrasian equilibria, we give polynomial-time algorithms that always find envy-free optimal dynamic prices in the remaining three cases. In contrast, for revenue maximization, we show that the corresponding problems are APX-hard if the ordering of the agents is fixed. On the positive side, we give polynomial-time algorithms for the setting when the seller can choose the order in which agents arrive.  

\medskip


\end{abstract}

\section{Introduction}
\label{sec:intro}

The great availability of surveys and marketing researches help sellers to predict the interest of costumers in different products, which opens up the possibility to apply user specific pricing processes. While this helps customers in purchase decisions, it makes the pricing process even more challenging for sellers. In this context, some simple economic rules are straightforward: low prices typically attract more customers but have a small revenue per item, while high prices generate a greater revenue per item but attract less customers. As sellers aim at maintaining customers' satisfaction while achieving high revenue, the need for effective pricing strategies is increasing.

We consider combinatorial markets, in which a set of indivisible items is to be distributed among a set of agents. Each agent has a valuation for each subset of items that measures how much receiving the bundle would be worth to the agent. An allocation assigns a subset of items to each agent so that every item is assigned to at most one of them. In a posted price mechanism, the seller can set the prices of the items individually, and the utility of an agent for a given bundle of items is the agent's value for the bundle minus the total price of all contained items -- the utility hence measures the agent's happiness when buying all the items of the bundle at the given prices. An allocation is considered to be envy-free if no agent would prefer to be assigned a different bundle of items.

In this paper, we study resource allocation problems in a dynamic environment from two perspectives. First, we focus on how to set prices so that a market equilibrium maximizes the overall social welfare, that is, the total sum of the agents' values. Second, we consider how to set prices that the seller's profit is maximized. To the best of our knowledge, the present work is the first one extending the concept of envy-freeness to the dynamic setting. 

\paragraph{Previous work.} 

Achieving optimal social welfare through simple mechanisms has been the center of attention for a long time due to its far-reaching applications. In particular, posted price mechanisms became a key approach to allocate resources, hence finding optimal pricing schemes is a fundamental question in combinatorial markets. A pair of pricing and allocation is called a Walrasian equilibrium if all the items are assigned to someone and each agent receives a bundle that maximizes her utility -- the definition automatically implies that the corresponding allocation maximizes social welfare. The idea of Walrasian equilibria was introduced already in the late 1800s~\cite{walras1874lausanne} and the existence of such an equilibria was verified for gross substitutes valuations by Kelso and Crawford~\cite{kelso1982job}. However, it was pointed out by Cohen-Addad et al. \cite{cohen2016invisible} and independently by Hsu et al. \cite{hsu2016prices} that the existence of Walrasian allocations strongly depends on the tie breaking process, usually carried out by a central coordinator. If the agents arrive one after the other and choose an arbitrary bundle of items that maximizes their utility, then the absence of a tie-breaking rule may result in a suboptimal allocation with respect the social welfare. To overcome these difficulties, Cohen-Addad et al. \cite{cohen2016invisible} introduced the notion of dynamic prices that proved to be a powerful tool in designing markets without a central tie-breaking coordinator. In the proposed model, agents arrive in an unspecified sequential order and the item prices can be updated before the arrival of the next agent. Their main result is a polynomial time dynamic pricing scheme that achieves optimal social welfare in unit-demand markets. This work initiated the study of dynamic pricing schemes, and the existence of optimal dynamic prices was settled for three agents with multi-demand valuations by Berger, Eden and Feldman~\cite{berger2020power}, for bi-demand valuations by B\'erczi, B\'erczi-Kov\'acs and Sz\"ogi~\cite{berczi2021dual}, for two agents with certain matroid rank valuations by B\'erczi, Kakimura and Kobayashi~\cite{berczi2021market}, and recently for four agents with multi-demand valuations by Pashkovich and Xie~\cite{pashkovich2022two}.

The market clearing condition can lead to Walrasian equilibrium with low revenue for the seller, even if prices are as high as possible. In the seminal work of Guruswami et al.~\cite{guruswami2005profit}, envy-free pricing was introduced as a relaxation of Walrasian equilibrium by dropping the requirement on the clearance of the market, but keeping the same fairness condition. In their model, the goal is to maximize the revenue when each agent has a demand and valuation for each bundle of items and each item has limited supply. The authors showed that maximizing the revenue is APX-hard already for unit-demand markets, and provided logarithmic approximations in the number of customers for the unit-demand and single-parameter cases. A similar hardness was proved independently by Aggarawal et al.~\cite{aggarwal2004algorithms}. Subsequently several versions of the problem have been shown to have poly-logarithmic inapproximability, see e.g. the works of Briest~\cite{briest2008uniform} and Chalmersook, Laekhanukit and Nanogkai~\cite{chalermsook2013independent}. Bansal et al. \cite{bansal2010dynamic} adapted the concept of envy-freeness to a pricing over time scheme. In their model, there is a single item with unlimited supply, and each agent is associated with a time interval over which she will consider buying a copy of the item, together with a maximum value the agent is willing to pay for it. The seller's goal is to set the prices at every time unit so as to maximize revenue from the sale of copies of the item over the time period.
 
\paragraph{Our results.}

The original motivation behind dynamic pricing schemes was to shift the tie-breaking process from the central coordinator to the customers, as in reality customers choose bundles of items without caring about social optimum. As it is shown by the above mentioned results, the dynamic setting is indeed capable of maximizing social welfare without the need for a central coordinator. On the other hand, this approach has an implication on the fairness of the final allocation that is usually not emphasized. The model assumes that the customers' sole objective is to pick a bundle of items maximizing their utility with respect to the prices available at their arrival, and they are not concerned with prices at earlier and/or later times. This means that envy-freeness is ensured only locally, and the final allocation together with the prices at which the items were bought do not necessarily form an envy-free solution over all time horizon. 

Our first contribution is initiating the study of dynamic pricing schemes under global fairness constraints. We extend the concept of envy-freeness to the dynamic setting in unit-demand markets, proposing four possible notions of different strength depending on whether the agents are concerned about the prices throughout entire time horizon, only in the past, only in the future, or only at their arrival. Note that the last case corresponds to the standard setting of dynamic pricing problems. We prove that, while ensuring envy-freeness for the entire time horizon basically brings the problem back to the case of static prices, the optimum social welfare can be achieved through envy-free dynamic prices in the remaining cases.

When it comes to revenue maximization, the results on the poly-logarithmic inapproximability of the optimal profit call for new approaches. Based on the success of dynamic pricing schemes in social welfare maximization, a natural idea is to combine the dynamic model with revenue maximization. The work of Bansal et al.~\cite{bansal2010dynamic} proposes a setup that resembles this idea. However, their model has a single item having unlimited supply, and agents are sold the item at the minimum price during their bid interval, which results in an allocation that is, again, envy-free only locally. 

Our second contribution is the analysis of the revenue maximization problem in a dynamic setting where fairness is defined using one of the above mentioned four possibilities. We show that, in contrast to welfare maximization, the flexibility of dynamic prices does not help in this case, and hence most of the problems are APX-hard.

Most of previous work on dynamic pricing schemes assumed that the customers arrive one after the other in an unspecified order. Apart from this standard case, we also consider two further options in all the above mentioned scenarios: when the customers arrive in a predetermined order, and when the seller has the opportunity to determine their order.

\paragraph{Practical motivation.}

Each notion of envy-freeness considered in his paper represents a natural concept of fairness that appears in everyday life. In certain markets, agents who arrive late do not perceive unfairness for prices that were posed earlier. In such situations, the seller tries to ensure that agents are not penalised by arriving too early. In other cases, discounting prices over time is a common strategy to guarantee clearance of the market. In such markets, agents are more inclined to accept that prices are low when the overall stock is low, and hence not to perceive an unfair pricing procedure. These are reasonable assumptions on customer behaviour and are resembled by our notion of envy-freeness when agents are concerned about prices only in the future and past, respectively.

In various market scenarios customers are obliged to register their arrival and interest in certain items. These cases can be further differentiated depending on whether buyers register to free time slots or they are assigned by the seller. Accordingly, the proposed assumptions on the arrival of agents is a realistic problem that sellers are facing.

\medskip
The rest of the paper is organized as follows. Basic notation and definitions are given in Section~\ref{sec:preliminaries}. Social welfare maximization is considered in Section~\ref{sec:welfare}. The main results of the paper are presented in Sections~\ref{sec:ex-post_welfare} and~\ref{sec:ex-ante_welfare}, where we give polynomial-time algorithms for social welfare maximization using ex-post and ex-ante envy-free dynamic prices, respectively. Results on revenue maximization are discussed in Section~\ref{sec:revenue}. Finally, in Section~\ref{sec:conclusions} we summarize the paper and list open problems that are subject of future research.

\section{Preliminaries}
\label{sec:preliminaries}

\paragraph{Basic notation.} 

We denote sets of \emph{real} and \emph{non-negative real numbers} by $\mathbb{R}$ and $\mathbb{R}_+$, respectively. For a positive integer $t$, we use $[t]$ to denote the set $\{1,\dots,t\}$. Let $\setIt$ be a ground set. For two subsets $X,Y\subseteq\setIt$, their \emph{symmetric difference} is $X\triangle Y\coloneqq (X\setminus Y)\cup(Y\setminus X)$. When $Y$ consists of a single element $y$, then the \emph{difference} $X\setminus \{y\}$ and \emph{union} $X\cup\{y\}$ are abbreviated by $X-y$ and $X+y$, respectively. For a function $f\colon \setIt\to\mathbb{R}$, the total sum of its values over $X$ is denoted by $f(X)\coloneqq \sum_{s\in X} f(s)$. For $X=\emptyset$, we define $f(\emptyset)=0$.

\paragraph{Graphs.} 

We denote a \emph{bipartite graph} by $G =(\setIt,\setPl;E)$, where $\setIt$ and $\setPl$ are the vertex classes and $E$ is the set of edges. By \emph{edge-weights}, we mean a function $w\colon E \to \mathbb{R}_{+}$. For a subset $X\subseteq\setIt\cup\setPl$, the \emph{subgraph of $G$ induced by $X$} is the graph obtained from $G$ by deleting all the vertices not contained in $X$, together with edges incident to them. We denote an edge of the graph going between $a\in\setPl$ and $i\in\setIt$ by $ai$. 

By orienting the edges of a bipartite graph, we get a \emph{directed graph} $D=(\setIt,\setPl;F)$, where $F$ is the set of arcs. A directed graph is called \emph{strongly connected} if every vertex is reachable from every other vertex through a directed path. A \emph{strongly connected component} of a directed graph is a subgraph that is strongly connected and is maximal with respect this property. By contracting each strongly connected component of a directed graph to a single vertex, one obtains an acyclic directed graph. Therefore, the strongly connected components have a so-called \emph{topological ordering} in which every arc going between components goes from an earlier component to a later one. 

\paragraph{Market model.}

A combinatorial market consists of a set $\setIt$ of \emph{indivisible items} and a set $\setPl$ of \emph{agents}. Throughout the paper, we denote by $\nIt\coloneqq|\setIt|$ and $\nPl\coloneqq|\setPl|$ the numbers of items and agents, respectively. An \emph{allocation} $\mathbf{X}$ assigns each agent $a$ a subset $X_a$ of items so that each item is assigned to at most one agent. 

In a \emph{unit-demand market}, each agent $\player \in \setPl$ has a valuation  $v_{\player}\colon\setIt\to \mathbb{R}_+$ over individual items and she desires only a single good, that is, we consider allocations $\mathbf{X}$ with $|X_a|\leq 1$ for $a\in\setPl$ -- in such cases we denote the item obtained by agent $a$ by $x_a$. We always assume that the agents' valuations are known in advance. Furthermore, we assume that $v_a(\emptyset) = 0$ for all agents $a \in \setPl$. Given prices $p(i)$ for each item $i\in\setIt$, the \emph{utility} of agent $a$ for item $i$ is $u_a(i)\coloneqq v_a(i)-p(i)$.  Then the \emph{social welfare} corresponding to the allocation is $\sum_{a\in\setPl} v_a(x_a)$, while the \emph{revenue} of the seller is $\sum_{a\in\setPl}p(x_a)$. 

In a \emph{static pricing scheme}, the seller sets the price $p(i)$ of each item $i\in\setIt$ in advance. Two fundamental problems in combinatorial markets are to find a pair of pricing vector $p\colon\setIt\to\mathbb{R}_+$ and allocation $\mathbf{X}$ such that the social welfare or the revenue is maximized. In contrast, in a \emph{dynamic pricing scheme} the agents arrive one after the other, and the seller can update the prices between their arrivals based on the remaining sets of items and agents. The order in which agents arrive is represented by a bijection $\sigma\colon\setPl\to[\nPl]$. The sets of agents, items and prices available before the arrival of the $t$th agent are denoted by $\setPl_t$, $\setIt_t$ and $p_t$, respectively. The utility of agent $\player$ for item $i$ at time step $t$ is then defined as $u_{\player,t}(i)\coloneqq  v_{\player}(i) - p_t(i)$. The next agent always chooses an item that maximizes her utility. After the last buyer has left, the pricing scheme terminates and results in pricing vectors $\mathbf{p}=(p_1, \ldots, p_\nPl)$ and an allocation $\mathbf{X}=(x_1,\ldots, x_\nPl)$, where $p_t$ is the price vector available at the arrival of the $t$th agent and $x_t$ is the item allocated to her. Note that $x_t$ might be an empty set if the utility of the agent is non-positive for each item in $\setIt_t$. We call a dynamic pricing scheme \emph{optimal} if the final allocation maximizes the objective, that is, the social welfare or the revenue, irrespective of the order in which the agents arrived.

In what follows, we define different variants of the model. Depending on whether ties between items are broken by the seller or the agents, we distinguish two cases:
\begin{itemize}\itemsep0em
  \item[(C1)] \emph{Seller-chooses.} If there are several items maximizing the utility of the current agent, then the seller decides which one to allocate to her.
  \item[(C2)] \emph{Agent-chooses.} If there are several items maximizing the utility of the current agent, then she decides which one to take.
\end{itemize}
In terms of finding an optimal pricing, problem (C1) is easier. Indeed, given an optimal pricing for (C2), the seller can always decide to allocate the item that was chosen by the agent.

Previous works generally assumed that agents arrive in an unspecified order. Besides this, we consider two further variants based on the control and information of the arrival process: 
\begin{itemize}\itemsep0em
  \item[(O1)] \emph{Unspecified.} The agents arrive in a fixed order that the seller has no information on.
  \item[(O2)] \emph{Predetermined.} The agents arrive in a fixed order that the seller knows in advance.
  \item[(O3)] \emph{Alterable.} The order of the agents is determined by the seller.
\end{itemize}  

Our model differs from earlier ones mainly in that we are seeking for optimal pricing schemes under fairness constraints. In the static setting, a pair of pricing $p$ and allocation $\mathbf{x}$ is \emph{envy-free} if $x_a\in\argmax\{u_a(i)\mid i\in\setIt\}$ holds for each agent $a\in\setPl$. The dynamic setting naturally suggests variants in which envy-freeness is defined over a subset of time steps. Let $T_a\subseteq [\nPl]$ be a subset of time steps for each agent $a\in\setPl$. Then price vectors $\mathbf{p}=(p_1, \ldots, p_\nPl)$ and allocation $\mathbf{X}=(x_1,\ldots, x_\nPl)$ form an envy-free allocation if $x_{a} \in \argmax\{u_{a,t}(i)\mid t\in T_a, i\in \setIt_t\}$ for each agent $\player\in\setPl$. We propose four possible notions of envy-freeness of different strength depending on the time period over which agents compare themselves to others:
\begin{itemize}\itemsep0em
  \item[(F1)] \emph{Strong envy-freeness.} Agents consider prices for the whole time horizon, that is, $T_a=\{1,\dots,\nPl\}$ for $a\in\setPl$. 
  \item[(F2)] \emph{Ex-post envy-freeness.} Agents consider prices available after and at their arrival, that is, $T_a=\{\sigma(a),\dots,\nPl\}$ for $a\in\setPl$.
  \item[(F3)] \emph{Ex-ante envy-freeness.} Agents consider prices available before and at their arrival, that is, $T_a=\{1,\dots,\sigma(a)\}$ for $a\in\setPl$.
  \item[(F4)] \emph{Weak envy-freeness.} Agents consider prices at their arrival, that is, $T_a=\{\sigma(a)\}$ for $a\in\setPl$.
\end{itemize}
Using this terminology, optimal dynamic pricing schemes discussed in \cite{cohen2016invisible,berger2020power,berczi2021dual,berczi2021market,pashkovich2022two} provide weakly envy-free solutions. It is worth mentioning that, though at first sight they might seem to be symmetric, the ex-post and ex-ante cases turns out to behave quite differently. 
  
As for the \emph{objective function}, we either consider the \emph{social welfare} $W(\mathbf{X})=\sum_{a\in\setPl} v_a(x_a)$ or the \emph{revenue} of all sold items $R(\mathbf{p},\mathbf{X}) = \sum_{a\in\setPl} p_{\sigma(a)}(x_a)$. 

These variants and the results presented in the paper are summarized in Table~\ref{table:results}. The results are split horizontally by the type of envy-freeness considered, while the columns are indexed by the type of the ordering of the agents. Algorithmic results hold irrespective of how agents break ties, while hardness results hold even if ties are broken by the seller. It is worth noting that the $O(\log(n))$-approximation algorithm of Guruswami et al.~\cite{guruswami2005profit} extends to all of variants of envy-free pricing where the objective is to maximize the revenue.

\begin{table}[t!]
\scriptsize
\centering
  \renewcommand*{\arraystretch}{1.3}
  \setlength{\tabcolsep}{3pt}
  \newcommand{\tableentrybase}[2]{\colorbox{#1}{\parbox[c][3.5em][c]{18mm}{\centering\scriptsize #2}}}
  \newcommand{\tableentrybasee}[2]{\colorbox{#1}{\parbox[c][1.5em][c]{18mm}{\centering\scriptsize #2}}}
  \newcommand{\tableentrybaseee}[2]{\colorbox{#1}{\parbox[c][3.5em][c]{19mm}{\centering\scriptsize #2}}}
  \newcommand{\tableentrybaseeee}[2]{\colorbox{#1}{\parbox[c][1.5em][c]{10mm}{\centering\scriptsize #2}}}
  \newcommand{\tableentrybaseeeee}[2]{\colorbox{#1}{\parbox[c][3.5em][c]{10mm}{\centering\scriptsize #2}}}
  
  \caption{Complexity landscape of social welfare and revenue maximization under fairness constraints in unit-demand markets. Algorithmic results (green cells) hold even in the agent-chooses setting, while hardness results (red cells) hold already for the seller-chooses case. In each row, complexities of cells with light shade are implied by cells with darker shade. \label{table:results}}
  \begin{tabular}{cc|ccc|ccc}
               & & \multicolumn{3}{c}{Welfare maximization} & \multicolumn{3}{c}{Revenue maximization}\\ 
             & Ties & Unspecified & Predetermined & Alterable & Unspecified & Predetermined & Alterable\\ \cmidrule(lr){1-2}\cmidrule(lr){3-3}\cmidrule(lr){4-4}\cmidrule(lr){5-5}\cmidrule(lr){6-6}\cmidrule(lr){7-7}\cmidrule(lr){8-8}\\[-1em]
            \rotatebox[origin=c]{90}{Strong} & \tableentrybaseeeee{ACMRed!0}{\tableentrybaseeee{ACMRed!0}{Agents} \linebreak \tableentrybaseeee{ACMRed!0}{Seller}} & \tableentrybaseee{ACMRed!0}{\tableentrybasee{ACMRed!20}{Not exists}\linebreak \tableentrybasee{ACMGreen!50}{P \linebreak \cite{walras1874lausanne,kelso1982job}}} & \tableentrybaseee{ACMRed!0}{\tableentrybasee{ACMRed!20}{Not exists}\linebreak \tableentrybasee{ACMGreen!20}{P}} & \tableentrybaseee{ACMRed!0}{\tableentrybasee{ACMRed!50}{Not exists \linebreak Rem. \ref{rem:strong_welfare}}\linebreak \tableentrybasee{ACMGreen!20}{P}} & \tableentrybaseee{ACMRed!0}{\tableentrybasee{ACMRed!20}{Not exists}\linebreak \tableentrybasee{ACMRed!20}{APX-hard}} & \tableentrybaseee{ACMRed!0}{\tableentrybasee{ACMRed!20}{Not exists}\linebreak \tableentrybasee{ACMRed!20}{APX-hard}} & \tableentrybaseee{ACMRed!0}{\tableentrybasee{ACMRed!50}{Not exists \linebreak Rem. \ref{rem:strong_welfare}}\linebreak \tableentrybasee{ACMRed!50}{APX-hard \linebreak Thm. \ref{thm:strong_revenue}}} \\ [2em]
            \rotatebox[origin=c]{90}{Ex-post} & & \tableentrybaseee{ACMGreen!0}{\tableentrybase{ACMGreen!50}{P \linebreak Thm.~\ref{thm:ex-post_welfare}}} & \tableentrybaseee{ACMGreen!0}{\tableentrybase{ACMGreen!20}{P}} & \tableentrybaseee{ACMGreen!0}{\tableentrybase{ACMGreen!20}{P}} & \tableentrybaseee{ACMRed!0}{\tableentrybase{ACMRed!20}{APX-hard}} & \tableentrybaseee{ACMRed!0}{\tableentrybase{ACMRed!50}{APX-hard  \linebreak Thm.~\ref{thm:ex_revenue_hard}}} & \tableentrybaseee{ACMRed!0}{\tableentrybase{ACMGreen!50}{P \linebreak Thm.~\ref{thm:ex-post_revenue_p}}}\\ [2em]
            \rotatebox[origin=c]{90}{Ex-ante} & & \tableentrybaseee{ACMRed!0}{\tableentrybase{ACMGreen!50}{P \linebreak Thm.~\ref{thm:ex-ante_welfare}}} & \tableentrybaseee{ACMRed!0}{\tableentrybase{ACMGreen!20}{P}} & \tableentrybaseee{ACMRed!0}{\tableentrybase{ACMGreen!20}{P}} & \tableentrybaseee{ACMRed!0}{\tableentrybase{ACMRed!20}{APX-hard}} & \tableentrybaseee{ACMRed!0}{\tableentrybase{ACMRed!50}{APX-hard  \linebreak Thm.~\ref{thm:ex_revenue_hard}}} & \tableentrybaseee{ACMRed!0}{\tableentrybase{ACMGreen!50}{P \linebreak Thm.~\ref{thm:ex-ante_revenue_p}}}\\ [2em]
            \rotatebox[origin=c]{90}{Weak}  & & \tableentrybaseee{ACMRed!0}{\tableentrybase{ACMGreen!50}{P \linebreak \cite[Thm. 3.1]{cohen2016invisible}}} & \tableentrybaseee{ACMRed!0}{\tableentrybase{ACMGreen!20}{P}} & \tableentrybaseee{ACMRed!0}{\tableentrybase{ACMGreen!20}{P}} & \tableentrybaseee{ACMRed!0}{\tableentrybase{ACMRed!0}{Open}} & \tableentrybaseee{ACMRed!0}{\tableentrybase{ACMGreen!50}{P \linebreak Thm.~\ref{thm:weak_revenue}}} & \tableentrybaseee{ACMRed!0}{\tableentrybase{ACMGreen!20}{P}} 
  \end{tabular}
\end{table}

\begin{rem}\label{rem:strong_welfare}
    In strongly envy-free pricing models, optimizing with respect to social welfare or revenue may lead to non-deterministic solutions when ties are broken by agents. In such cases, `optimality' of a pricing scheme is not well-defined. This is well-illustrated by the classic example of Cohen-Addad et al. \cite{cohen2016invisible} with three items $i_1,i_2,i_3$ and three agents $a_1,a_2,a_3$ having valuations $v_{a_j}(i_{j})=v_{a_j}(i_{j+1})=1$ $v_{a_j}(i_{j+2})=0$ for $j\in[3]$, where indices are meant in a cyclic order. Since each item has the same value for two of the agents, strong envy-freeness implies that the seller should set the prices uniformly and cannot update them. Assume that $p(i_1) = p(i_2) = p(i_3) = 1$ and that agent $j_3$ arrives first who chooses item $i_3$. If $j_1$ arrives next, she is indifferent between $i_1$ and $i_2$, but the achieved social welfare or revenue heavily depends on her decision. 
    To overcome these difficulties, one could consider different objective function, such as worst-case or average social welfare or revenue. However, this lies beyond the scope of this paper and we postpone them as subjects of future research.
\end{rem}

\paragraph{Weighted coverings.}

A unit-demand combinatorial market can be represented by a complete edge-weighted bipartite graph $G=(\setIt,\setPl;E)$, where vertex classes $\setIt$ and $\setPl$ correspond to the sets of items and agents, respectively. For any item $i\in\setIt$ and agent $a\in\setPl$, the weight of the edge $ai$ is $w(ia)\coloneqq v_{a}(i)$. Then there is a one-to-one correspondence between allocations maximizing social welfare and maximum weight matchings of $G$. Even more, the maximum weight of a matching is clearly an upper bound on the maximum revenue achievable through any pricing mechanism. These observations motivate to investigate dynamic pricing schemes through the lenses of maximum weight matchings.

Let us denote the vertex set of $G$ by $V\coloneqq \setIt\cup\setPl$. A function $\pi\colon V \to \mathbb{R}$ is a \emph{weighted covering} if $\pi(i)+\pi(a) \geq w(ia)$ holds for every edge $ia \in E$. The \emph{total value} of the covering is $\pi(V)=\sum_{v \in V} \pi(v)$. A weighted covering of minimum total value is called \emph{optimal}. A cornerstone result of graph optimization is due to Egerváry~\cite{egervary1931matrixok} who provided a min-max characterization for the maximum weight of a matching in a bipartite graph.

\begin{theorem}[Egerv{\'a}ry~\cite{egervary1931matrixok}]\label{thm:egervary}
Let $G=( \setIt,\setPl ;E)$ be a bipartite graph and $w\colon E \to \mathbb{R}$ be a weight function on the set of edges. Then the maximum weight of a matching is equal to the minimum total value of a non-negative weighted covering $\pi$ of $w$. 
\end{theorem}

Given a weighted covering $\pi$, item $i\in\setIt$ and agent $a\in\setPl$, the edge $ai$ is called \emph{tight with respect to $\pi$} if $\pi(i)+\pi(a)=w(ia)$. The \emph{subgraph of tight edges} is then denoted by $G_\pi = ( \setIt, \setPl ;E_\pi)$. We call an edge $ia\in E$ \emph{legal} if there exists a maximum weight matching containing it, and in such a case we say that $i$ is \emph{legal} for $a$. It is known that legal edges are always tight with respect to any optimal weighted covering, while the converse does not always hold, that is, a tight edge is not necessarily legal. However,~\cite[Lemma 5]{berczi2021dual} showed that a careful choice of $\pi$ ensures the sets of tight and legal edges to coincide.

\begin{lemma}[Bérczi, Bérczi-Kovács and Szögi~\cite{berczi2021dual}]\label{lemma:dual2}
The optimal $\pi$ attaining the minimum in Theorem~\ref{thm:egervary} can be chosen such that 
\begin{itemize}
    \item[(a)] an edge $ai$ is tight with respect to $\pi$ if and only if it is legal, and 
    \item[(b)] $\pi(v)=0$ for some $v\in V$ if and only if there exists a maximum weight matching $M$ with $d_M = 0$.
\end{itemize}
Furthermore, such a $\pi$ can be determined in polynomial time.
\end{lemma}

Finally, we will use the following technical lemma, see~\cite[Lemma 1]{berczi2021dual}.

\begin{lemma}\label{lem:all}
Given a bipartite graph $G=(\setIt,\setPl;E)$ corresponding to unit-demand combinatorial market, we may assume that all items are covered by every maximum weight matching of $G$. 
\end{lemma}

\section{Maximizing the social welfare}
\label{sec:welfare}

Since a Walrasian equilibrium maximizes social welfare and ensures envy-freeness at the same time, the existence of optimal dynamic pricing schemes under fairness constraints is settled when ties are broken by the seller. The seminal paper of Cohen-Addad et al.~\cite{cohen2016invisible} initiated the study of the agent-chooses case, and provided an algorithm for determining a weakly envy-free solution in unit-demand markets. 

A different proof of the same result was later given by Bérczi, Bérczi-Kovács and Szögi~\cite{berczi2021dual}. Their algorithm starts with a minimum weighted covering provided by Lemma~\ref{lemma:dual2} in the edge-weighted bipartite graph representing the market, and sets the initial prices according to the covering values. As a result, the first agent $a$ chooses an item $i$ such that $ai$ is tight, hence $i$ is legal for $a$. Based on this, it may seem that, keeping the same prices throughout, the resulting allocation will eventually be a maximum weight matching. However, after item $i$ is taken, an edge that was legal before might become non-legal. To overcome this, the weighted covering needs to be updated in the remaining graph at each time step, which causes the price of an item fluctuating over time. For that reason, the algorithm does not extend to the ex-post and ex-ante envy-free cases. 

To prevent the fluctuation of prices, we do not recompute the weighted covering at each time step from scratch. Instead, we fix a single weighted covering at the very beginning, and then we always slightly modify it to control the agents' choices in such a way that 
\begin{enumerate}[label=(\Alph*)]\itemsep0em 
\item\label{it:a} no matter which agent arrives next, if she is covered by every maximum weight matching in the current graph then she picks an item that is legal for her, otherwise she either picks an item that is legal for her or does not take an item at all, and
\item\label{it:b} the price-changes from time step $t-1$ to $t$ are limited to non-increases in the ex-post and to non-decreases in the ex-ante case. 
\end{enumerate}
The first property implies that the final allocation corresponds to a maximum weight matching of $G$, hence it maximizes social welfare. The second property ensures that the resulting allocation meets the requirements of ex-post or ex-ante envy-freeness.

\subsection{Preparations}
\label{sec:preparartion}

Consider the edge-weighted bipartite graph $G=(\setIt,\setPl;E)$ representing the market. By Lemma~\ref{lem:all}, we may assume that all items are covered by every maximum weight matching of $G$. Take a weighted covering $\pi$ provided by Lemma~\ref{lemma:dual2}. Throughout this section, tightness of an edge is always meant with respect to $\pi$. Recall that $G_\pi$ denotes the subgraph of tight edges. 

At time step $t$, the sets of remaining agents and items are denoted by $\setPl_t$ and $\setIt_t$, respectively. We denote the subgraph of $G_\pi$ induced by vertices $\setIt_t\cup\setPl_t$ by $G_t=(V_t,E_t)$. For a maximum weight matching $M_t$ of $G_t$, let $x^t_a$ denote the item to which $a$ is matched in $M_t$ if such an item exists, otherwise define $x^t_a$ to be the empty set. Note that the edges $ax^t_a$ are obviously legal in $G_t$.

Given such a matching $M_t$, we construct a directed graph $D_t$ as follows. We add another copy of every edge in $M_t$ to $G_t$ that we refer to as dummy edges. Then we orient the original copies in $M_t$ from $\setIt_t$ to $\setPl_t$, and orient all the remaining edges -- including the dummy ones -- from $\setPl_t$ to $\setIt_t$. We denote the strongly connected components of the resulting directed graph by $C^t_1,\dots,C^t_{q_t}$ indexed according to a topological ordering, see Figure~\ref{fig:ex} for an example.

\begin{figure}[t!]
\centering
\begin{subfigure}[t]{0.32\textwidth}
  \centering
  \includegraphics[width=.8\linewidth]{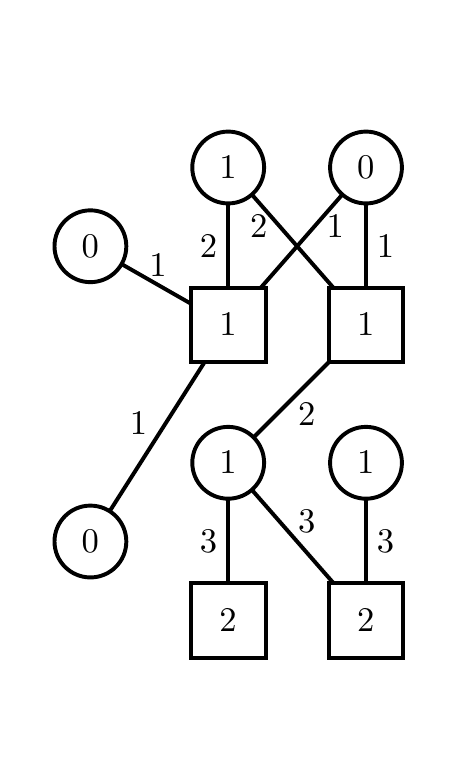}
  \caption{$G_t$ with edge-weights and weighted covering values.}
  \label{fig:ex-post1}
\end{subfigure}\hfill
\begin{subfigure}[t]{0.32\textwidth}
  \centering
  \includegraphics[width=.8\linewidth]{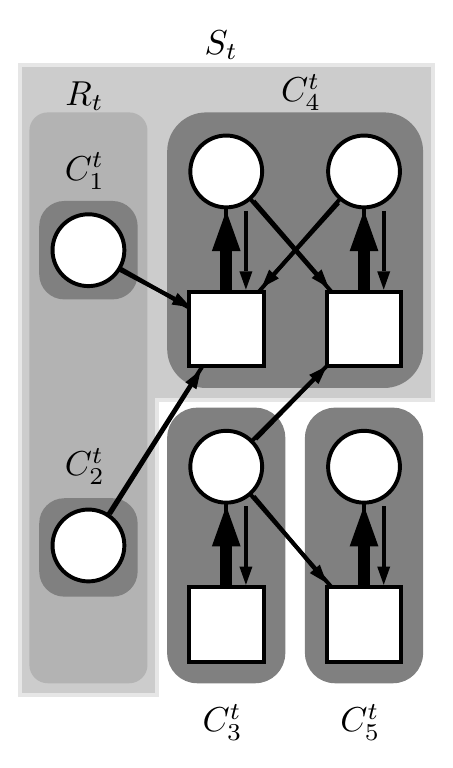}
  \caption{Definition of the set $S_t$ in the proof of Theorem~\ref{thm:ex-post_revenue_p}.}
  \label{fig:ex-post2}
\end{subfigure}\hfill
\begin{subfigure}[t]{0.32\textwidth}
  \centering
  \includegraphics[width=.8\linewidth]{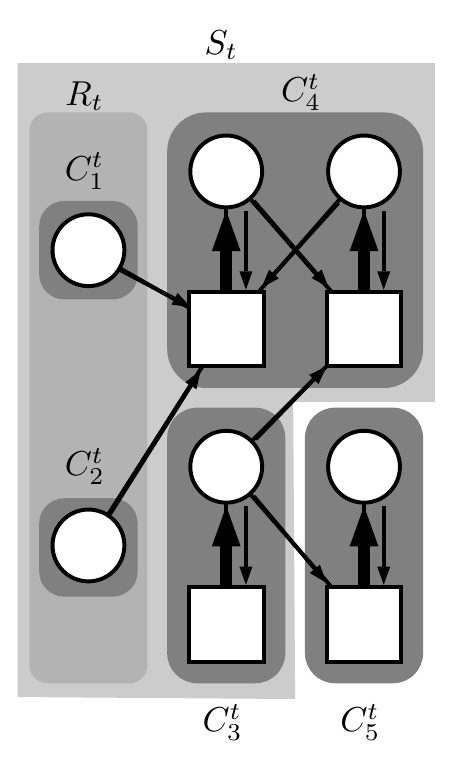}
  \caption{Definition of the set $S_t$ in the proof of Theorem~\ref{thm:ex-ante_revenue_p}.}
  \label{fig:ex-post3}
\end{subfigure}
\caption{Illustration of the constructions for the ex-post and ex-ante cases. Circles and squares correspond to agents and items, respectively. Thick edges denote a maximum weight matching $M_t$, while $C^t_1,\dots,C^t_5$ are the strongly connected components of the directed graph $D_t$.}
\label{fig:ex}
\end{figure}

The following technical claims will be used later.

\begin{claim}\label{cl:legal}
Let $a\in\setPl_t$ and $i\in\setIt_t$ be such that they are in the same strongly connected component of $D_t$ and $ai$ is tight. Then $i$ is legal for $a$ in $G_t$.
\end{claim}
\begin{proof}
If the edge is oriented from $i$ to $a$, then $ia\in M_t$ and the statement clearly holds. Otherwise, the edge is oriented from $a$ to $i$. Since every arc of a strongly connected digraph is contained in a directed cycle, it suffices to show that any directed cycle $C$ consists only of legal edges. This follows from the fact that all the edges of $G_t$ are tight, hence $M_t\triangle C$ is also a maximum weight matching of $G_t$, implying that it consists of legal edges.
\end{proof}

In both the ex-post and ex-ante cases, our proof builds on the following idea. 
It is not difficult to check that if one sets the price of each item to its $\pi$ value, then agents strictly prefer items to which they are connected through a tight edge over items for which the corresponding edge is not tight, see Lemma~\ref{cl:prefer}. However, as time passes, a tight edge is no longer necessarily legal. In order to prevent the next agent to choose such an edge, we will slightly change the prices in a case-specific way. To do this, let $\delta>0$ be a constant for which \[\delta<\frac{1}{2}\min\bigl\{\min\{\pi(a)+\pi(i)-v_a(i)\mid ia\in E\ \text{is not tight}\},\ \min\{\pi(i)\mid i\in \setIt\}\bigr\}.\] Note that such a $\delta$ exists by Lemmas~\ref{lemma:dual2}(b) and~\ref{lem:all}. We also choose a small constant $0<\varepsilon< \delta/(\nPl2^\nPl)$ that will be used later on. 
The next claim shows that tight edges lead to greater utility than non tight ones.

\begin{claim}\label{cl:prefer}
Let $a\in\setPl_t$ and $i,i'\in\setIt_t$ be such that $ai$ is tight, $ai'$ is not tight. If $p_t(i)\leq \pi(i)+\delta$ and $p_t(i')\geq \pi(i')-\delta$, then $a$ strictly prefers $i$ over $i'$.
\end{claim}
\begin{proof}
Assume that $p_t(i)\leq \pi(i)+\delta$ and $p_t(i')\geq \pi(i')-\delta$. Since $ai$ is tight and $ai'$ is not tight, we have $v_a(i)=\pi(a)+\pi(i)$ and $v_a(i')+2\delta<\pi(a)+\pi(i')$ by the definition of $\delta$. Thus we get
\begin{align*}
u_{a,t}(i')
{}&{}\leq
v_a(i')-(\pi(i')-\delta)\\
{}&{}<
\pi(a)-\delta\\
{}&{}=
v_{a}(i)-(\pi(i)+\delta)\\
{}&{}\leq
u_{a,t}(i).
\end{align*}
Hence the utility of $a$ for item $i$ is strictly larger than for item $i'$, concluding the proof.
\end{proof}

Finally, the following two technical claim follows easily from the definitions.

\begin{claim}\label{cl:nottightut}
Let $a\in\setPl_t$ and $i\in\setIt_t$ be such that $ai$ is not tight. If $p_t(i)\geq \pi(i)-\delta$, then $u_{a,t}(i)<\pi(a)$.
\end{claim}
\begin{proof}
As $ai$ is not tight, the definition of $\delta$ implies $u_{a,t}(i)=v_a(i)-p_t(i)< (\pi(a)+\pi(i)-\delta)-(\pi(i)-\delta)=\pi(a)$ as stated.
\end{proof}

\begin{claim}\label{cl:tightut}
Let $a\in\setPl_t$ and $i\in\setIt_t$ be such that $ai$ is tight. If $p_t(i)=\pi(i)+\omega$ for some real number $\omega$, then $u_{a,t}(i)=\pi(a)-\omega$. 
\end{claim}
\begin{proof}
As $ai$ is tight, we get $u_{a,t}(i)=v_a(i)-p_t(i)=(\pi(a)+\pi(i))-(\pi(i)+\omega)=\pi(a)-\omega$ as stated.
\end{proof}

The set of agents not matched in $M_t$ is denoted by $R_t\coloneqq\{a\in\setPl_t\mid x^t_a=\emptyset\}$. Note that $\pi(a)=0$ for each $a\in R_1$ by Lemma~\ref{lemma:dual2}(b). In fact, we will maintain this property for each set $R_t$ throughout the pricing process. 

From this point, we discuss the ex-post and ex-ante cases separately as their proofs differ in that prices must be set differently. 

\subsection{Ex-post envy-free pricing}
\label{sec:ex-post_welfare}

The algorithm in \cite{berczi2021dual} for the unit-demand case updates the prices at each time step using an arbitrary weighted covering in the remaining graph. In order to adapt a similar idea for the ex-post case, we do this in a controlled manner which ensures that prices do not increase over time.

\begin{theorem}\label{thm:ex-post_welfare}
There exists a welfare-maximizing dynamic pricing scheme for the unit-demand ex-post envy-free pricing problem even if ties are broken by the agents. Furthermore, the optimal prices can be determined in polynomial time.
\end{theorem}
\begin{proof}
Throughout the algorithm, we maintain a maximum weight matching $M_t$ in the graph $G_t$ of remaining agents and items. Recall that $R_t$ denotes the set of agents not covered by the matching $M_t$. Initially, $\pi(a)=0$ for $a\in R_1$ by Lemma~\ref{lemma:dual2}(b), and we will maintain this property for each $a\in R_t$ throughout. 

We describe a general phase $t\in[\nPl]$ of the pricing process. We define the set \[S_t\coloneqq\{v\in V_t\mid \text{there exists a directed path to $v$ from an agent $a\in R_t$}\},\] see Figure~\ref{fig:ex-post2} for an example. In particular, $R_t\subseteq S_t$. Note that either $C^t_j\subseteq S_t$ or $C^t_j\cap S_t=\emptyset$ for each strongly connected component $C^t_j$ of $D_t$. The prices are then updated as follows:
\begin{equation*}
    p_t(i)=\begin{cases}
        \pi(i)+\delta/2^t+j\varepsilon & \text{if $i\in C^t_j$ such that $C^t_j\subseteq S_t$},\\
        \pi(i)-\delta(1-1/2^t)+j\varepsilon & \text{if $i\in C^t_j$ such that $C^t_j\cap S_t=\emptyset$}.
    \end{cases}
\end{equation*}
By the choice of $\delta$ and $\varepsilon$, the prices are non-negative. Let $a\in\setPl_t$ denote the agent who arrives at time step $t$. The next three claims together show that the prices satisfy property \ref{it:a}.

\begin{claim}\label{cl:postnormal1}
If $a\in V_t\setminus S_t$, or $a\in S_t\setminus R_t$ and $\pi(a)>0$, then she chooses an item that is legal for her in $G_t$. 
\end{claim}
\begin{proof}
Since $a\in V_t\setminus R_t$, the matching $M_t$ covers $a$ and hence $x_a$ is non-empty. Let $i\in\setIt_t$ be an arbitrary item distinct from $x_a$. If $ai$ is not tight, then the agent strictly prefers $x_a$ over $i$ by Claim~\ref{cl:prefer}. If $ai$ is tight but $i$ is in a different strongly connected component than $a$, then the index of the component of $i$ is strictly larger than that of $a$. Furthermore, by the definition of the set $S_t$, either both of them are in $S_t$ or $a\notin S_t$ and $i\in S_t$. These together imply that $p_t(x_a)-\pi(x_a)<p_t(i)-\pi(i)$, hence the agent strictly prefers $x_a$ over $i$ by Claim~\ref{cl:tightut}. Finally, if $ai$ is tight and $i$ is in the same strongly connected component as $a$, then $ai$ is legal by Claim~\ref{cl:legal}. The utility of $a$ is non-negative for such items by the choice of $\delta$ and the assumptions on $a$, hence the claim follows.
\end{proof}

\begin{claim}\label{cl:postnormal2}
If $a\in S_t\setminus R_t$ and $\pi(a)=0$, then she takes no item at all and there exists a maximum weight matching in $G_t$ that does not cover $a$. 
\end{claim}
\begin{proof}
Let $i\in\setIt_t$ be an arbitrary item. If $ai$ is not tight, then the utility of $a$ for $i$ is negative by Claim~\ref{cl:nottightut}. If $ai$ is tight, then $i\in S_t$ by the definition of the set $S_t$. This implies $p_t(i)>\pi(i)$ by the definition of the prices, hence the utility of $a$ for $i$ is negative by Claim~\ref{cl:tightut}. However, in such cases there exists a directed path $P$ to $a$ from an agent $a'\in R_t$. The fact that $P$ consist of tight edges together with $\pi(a)=\pi(a')=0$ imply that $M_t\triangle P$ is also a maximum weight matching in $G_t$.
\end{proof}

\begin{claim}\label{cl:postnone}
If $a\in R_t$, then she takes no item at all.
\end{claim}
\begin{proof}
By assumption, we have $\pi(a)=0$. Let $i\in\setIt_t$ be an arbitrary item. If $ai$ is not tight, then the utility of $a$ for $i$ is negative by Claim~\ref{cl:nottightut}. If $ai$ is tight, then $i\in S_t$ by the definition of the set $S_t$. This implies $p_t(i)>\pi(i)$ by the definition of the prices, hence the utility of $a$ for $i$ is negative by Claim~\ref{cl:tightut}.
\end{proof}

The matching $M_t$, and implicitly the set $R_t$, is updated as follows. If $a\in V_t\setminus S_t$, or $a\in S_t\setminus R_t$ and $\pi(a)>0$, then $a$ takes an item $i$ from her strongly connected component. If $ai\in M_t$, then set $M_{t+1}\coloneqq M_t\setminus\{ai\}$. Otherwise, let $C$ be a directed cycle of $D_t$ containing the arc $ai$, and set $M_{t+1}\coloneqq (M_t\triangle C)\setminus\{ai\}$. In this case, we have $R_{t+1}=R_t$. If $a\in S_t$ and $\pi(a)=0$, then consider the directed path $P$ to $a$ from an agent $a'\in R_t$, and set $M_{t+1}\coloneqq (M_t\triangle P)\setminus\{ai\}$, implying $R_{t+1}=R_t\setminus\{a'\}$. Finally, if $a\in R_t$, then $M_{t+1}\coloneqq M_t$, hence $R_{t+1}=R_t$. 

It remains to verify that the pricing scheme satisfies property \ref{it:b}, which is done by the following statement.

\begin{claim}\label{cl:decrease}
The price of any item does not increase over time. 
\end{claim}
\begin{proof}
At each phase of the algorithm, the price of an item $i$ is obtained by shifting its original weighted covering value. Though the structure of the directed graph $D_t$ and therefore the index of the strongly connected component containing $i$ might change from phase to phase, the choice of $\varepsilon$ ensures that 
$\pi(i)+\delta/2^t>\pi(i)+\delta/2^{t+1}+n\varepsilon$ and $\pi(i)-\delta(1-1/2^t)>\pi(i)-\delta(1-1/2^{t+1})+n\varepsilon$. Hence, in order to verify the claim, it suffices to show that $S_{t+1}\subseteq S_t$. 

Note that $M_{t+1}$ is chosen in such a way that no arc of $D_{t+1}$ leaves the set $S_t\cap\setIt_{t+1}$. Indeed, $D_{t+1}$ is obtained from $D_t$ by possibly reorienting a directed cycle or a directed path that lies completely in $S_t$, and then deleting an agent and possibly an item. These steps do not result in a directed arc leaving $S_t\cap\setIt_{t+1}$, implying $S_{t+1}\subseteq S_t$.  
\end{proof}

By Claims~\ref{cl:postnormal1}-\ref{cl:postnone}, if the next agent is covered by the matching $M_t$ then she either chooses an item that is legal for her, or $M_{t+1}$ is also a maximum weight matching of $G_t$. Otherwise, she does not take any of the items. This implies that the resulting allocation corresponds to a maximum weight matching of $G$ and hence maximizes social welfare. By Claim~\ref{cl:decrease}, the prices do not increase over time. This implies that the solution is ex-post envy-free, concluding the proof of the theorem.  
\end{proof}

\subsection{Ex-ante envy-free pricing}
\label{sec:ex-ante_welfare}

To give an algorithm for the ex-ante case, we adopt the proof of Theorem~\ref{thm:ex-post_welfare}. However, to ensure that the final prices and allocation form an ex-ante envy-free solution, prices have to be updated differently.

\begin{theorem}\label{thm:ex-ante_welfare}
There exists a welfare-maximizing dynamic pricing scheme for the unit-demand ex-ante envy-free pricing problem even if ties are broken by the agents. Furthermore, the optimal prices can be determined in polynomial time.
\end{theorem}
\begin{proof}   
Similarly to the ex-post case, we maintain a maximum weight matching $M_t$ in the graph $G_t$ of remaining agents and items, and denote by $R_t$ the set of agents not covered by the matching $M_t$. Initially, $\pi(a)=0$ for $a\in R_1$ by Lemma~\ref{lemma:dual2}(b), and we will maintain this property for each $a\in R_t$ throughout. 

We describe a general phase $t\in[\nPl]$ of the pricing process. We define the set \[S_t\coloneqq\{v\in V_t\mid \text{there exists a directed path from $v$ to an agent $a\in\setPl_t\setminus R_t$ with $\pi(a)=0$}\},\]see Figure~\ref{fig:ex-post3} for an example. Note that either $C^t_j\subseteq S_t$ or $C^t_j\cap S_t=\emptyset$ for each strongly connected component $C^t_j$ of $D_t$. The prices are then updated as follows:
\begin{equation*}
    p_t(i)=\begin{cases}
        \pi(i)-\delta/2^t+j\varepsilon & \text{if $i\in C^t_j$ such that $C^t_j\subseteq S_t$},\\
        \pi(i)+\delta(1-1/2^t)+j\varepsilon & \text{if $i\in C^t_j$ such that $C^t_j\cap S_t=\emptyset$}.
    \end{cases}
\end{equation*}
By the choice of $\delta$ and $\varepsilon$, the prices are non-negative. Let $a\in\setPl_t$ denote the agent who arrives at time step $t$. The next three claims together show that the prices satisfy property \ref{it:a}.

\begin{claim}\label{cl:normal}
If $a\in V_t\setminus R_t$, then she chooses an item that is legal for her in $G_t$. 
\end{claim}
\begin{proof}
Since $a\in V_t\setminus R_t$, the matching $M_t$ covers $a$ and hence $x_a$ is non-empty. Let $i\in\setIt_t$ be an arbitrary item distinct from $x_a$. If $ai$ is not tight, then the agent strictly prefers $x_a$ over $i$ by Claim~\ref{cl:prefer}. If $ai$ is tight but $i$ is in a different strongly connected component than $a$, then 
the index of the component of $i$ is strictly larger than that of $a$. Furthermore, by the definition of the set $S_t$, either both or none of them are contained in $S_t$. These together imply that $p_t(x_a)-\pi(x_a)<p_t(i)-\pi(i)$, hence the agent strictly prefers $x_a$ over $i$ by Claim~\ref{cl:tightut}. Finally, if $ai$ is tight and $i$ is in the same strongly connected component as $a$, then $ai$ is legal by Claim~\ref{cl:legal}. The utility of $a$ is non-negative for such items by the choice of $\delta$, hence the claim follows.
\end{proof}

\begin{claim}\label{cl:none}
If $a\in R_t\setminus S_t$, then she takes no item at all.
\end{claim}
\begin{proof}
By assumption, we have $\pi(a)=0$. Let $i\in\setIt_t$ be an arbitrary item. If $ai$ is not tight, then the utility of $a$ for $i$ is negative by Claim~\ref{cl:nottightut}. If $ai$ is tight, then $i\notin S_t$ by the definition of the set $S_t$. This implies $p_t(i)>\pi(i)$ by the definition of the prices, hence the utility of $a$ for $i$ is negative by Claim~\ref{cl:tightut}.
\end{proof}

\begin{claim}\label{cl:path}
If $a\in R_t\cap S_t$, then she either chooses an item that is legal for her in $G_t$, or takes no item at all.
\end{claim}
\begin{proof}
By assumption, we have $\pi(a)=0$. Let $i\in\setIt_t$ be an arbitrary item. If $ai$ is not tight, then the utility of $a$ for $i$ is negative by Claim~\ref{cl:nottightut}. If $ai$ is tight but $i\notin S_t$, then $p_t(i)>\pi(i)$ by the definition of the prices, hence the utility of $a$ for $i$ is negative by Claim~\ref{cl:tightut}. If $ai$ is tight and $i\in S_t$, then there exists a directed path $P$ from $a$ to an agent $a'$  in $D_t$ which is covered by $M_t$ and $\pi(a')=0$. The fact that $P$ consists of tight edges together with $\pi(a)=\pi(a')=0$ imply that $M_t\triangle P$ is also a maximum weight matching in $G_t$, there fore $i$ is legal for $a$. 
\end{proof}

The matching $M_t$, and implicitly the set $R_t$, is updated as follows. If $a\in V_t\setminus R_t$, then $a$ takes an item $i$ from her strongly connected component. If $ai\in M_t$, then set $M_{t+1}\coloneqq M_t\setminus\{ai\}$. Otherwise, let $C$ be a directed cycle of $D_t$ containing the arc $ai$, and set $M_{t+1}\coloneqq (M_t\triangle C)\setminus\{ai\}$. In this case, we have $R_{t+1}=R_t$. If $a\in R_t\setminus S_t$ or $a\in R_t\cap S_t$ but $a$ takes no item, then set $M_{t+1}\coloneqq M_t$, implying $R_{t+1}=R_t\setminus\{a\}$. Finally, if $a\in R_t\cap S_t$ and $a$ takes an item $i$, then consider the directed path $P$ from $a$ to an agent $a'$ which is covered by $M_t$ and $\pi(a')=0$, and set $M_{t+1}\coloneqq (M_t\triangle P)\setminus\{ai\}$. Since we have $R_{t+1}=R_t\setminus\{a\}\cup\{a'\}$ and $\pi(a)=0$, the property that each agent in $R_{t+1}$ has $\pi$ value $0$ holds.

It remains to verify that the pricing scheme satisfies property \ref{it:b}, which is done by the following statement.

\begin{claim}\label{cl:increase}
The price of any item does not decrease over time. 
\end{claim}
\begin{proof}
At each phase of the algorithm, the price of an item $i$ is obtained by shifting its original weighted covering value. Though the structure of the directed graph $D_t$ and therefore the index of the strongly connected component containing $i$ might change from phase to phase, the choice of $\varepsilon$ ensures that 
$\pi(i)-\delta/2^t+n\varepsilon<\pi(i)-\delta/2^{t+1}$ and $\pi(i)+\delta(1-1/2^t)+n\varepsilon<\pi(i)+\delta(1-1/2^{t+1})$. Hence, in order to verify the claim, it suffices to show that $S_{t+1}\subseteq S_t$. 

Note that $M_{t+1}$ is chosen in such a way that no arc of $D_{t+1}$ enters the set $S_t\cap\setIt_{t+1}$. Indeed, $D_{t+1}$ is obtained from $D_t$ by possibly reorienting a directed cycle or a directed path that lies completely in $S_t$, and then deleting an agent and possibly an item. These steps do not result in a directed arc entering $S_t\cap\setIt_{t+1}$, implying $S_{t+1}\subseteq S_t$.     
\end{proof}

By Claims~\ref{cl:normal}-\ref{cl:path}, if the next agent is covered by the matching $M_t$ then she chooses an item that is legal for her. Otherwise, she either chooses an item that is legal for her, or does not take any of the items. This implies that the resulting allocation corresponds to a maximum weight matching of $G$ and hence maximizes social welfare. By Claim~\ref{cl:increase}, the prices do not decrease over time. This implies that the solution is ex-ante envy-free, concluding the proof of the theorem.  
\end{proof}

\section{Maximizing the revenue}
\label{sec:revenue}

When it comes to revenue maximization in the static setting, the problem is not only hard to solve but also to approximate within a reasonable factor, and the difficulty stems from the lack of strong upper bounds. Clearly, the maximum weight of a matching is an upper bound on the total revenue achievable through pricing mechanisms, but the gap between optimal revenue and maximum weight of a matching may be $O(\log(n))$. Indeed, consider a market with items $\setIt=\{i_1,\dots,i_\nPl\}$ and agents $\setPl=\{a_1,\dots,a_\nPl\}$. Let the valuations be defined as $v_{a_j}(i_k) \coloneqq 1/j$ for $1\leq j\leq \nPl$ and $ j\leq k\leq\nIt$ and $0$ otherwise, see Figure~\ref{fig:harm} for an illustration. Then there is a unique maximum weight matching between agents and items that consists of the edges $i_ja_j$ for $1\leq j\leq n$ with total weight $\sum_{j=1}^n 1/j$. For any pair of envy-free static pricing and allocation, if an agent $a_j$ receives an item $i_k$ at some price $p(i_k)$, then the price of all the other items must be at least $p(i_k)$ to ensure that agent $a_j$ is not envious. On the other hand, the price $p(i_k)$ cannot be greater than $1/j$ as otherwise the utility of agent $a_j$ for item $i_k$ is negative. These observations together imply that the price of each item sold is at most $1/j$ where $j$ is the largest index for which agent $a_j$ receives an item. Hence the total revenue is at most $j\cdot1/j=1$, leading to an $O(\log(n))$ gap as stated.

\begin{figure}[t!]
\centering
\begin{subfigure}[t]{0.47\textwidth}
  \centering
  \includegraphics[width=.8\linewidth]{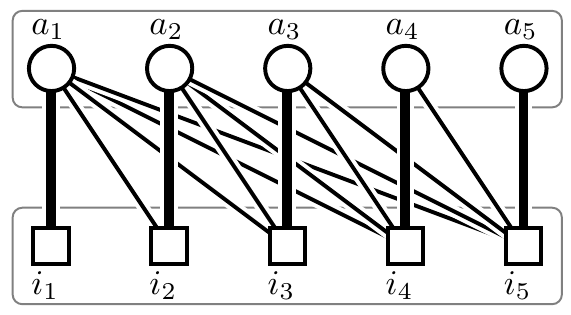}
  \caption{The weight of an edge $a_ji_k$ is $1/j$, while missing edges have weight $0$. Thick edges correspond to the unique maximum weight matching of total weight $\sum_{j=1}^n 1/j$.}
  \label{fig:harmonic}
\end{subfigure}\hfill
\begin{subfigure}[t]{0.47\textwidth}
  \centering
  \includegraphics[width=.8\linewidth]{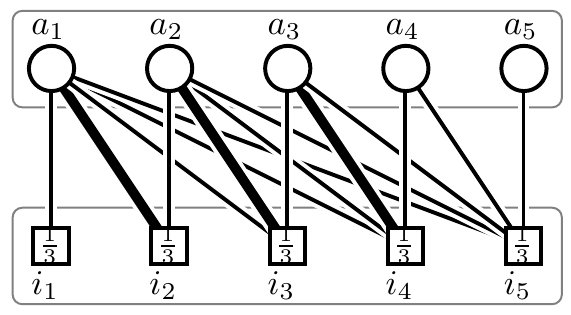}
  \caption{Values on items and thick edges correspond to a pair of envy-free pricing and allocation, respectively. The revenue is always less than or equal to $1$.}
  \label{fig:harmonicb}
\end{subfigure}
\caption{Illustration of $O(\log(n))$ gap between the optimal revenue and maximum weight of a matching in the case of envy-free static pricing.}
\label{fig:harm}
\end{figure}

In what follows, we turn our attention to the revenue maximization problem in a dynamic environment with fairness constraints.

\subsection{Hardness results}
\label{sec:hardness}

When the solution is required to be strongly envy-free, then the dynamic setting does not make a difference compared to the static one, as shown by the following theorem.

\begin{theorem}\label{thm:strong_revenue}
Maximizing the revenue in the unit-demand strongly envy-free dynamic pricing problem is APX-hard, even if the agents' ordering is chosen and ties are broken by the seller.
\end{theorem}
\begin{proof}
Let $\sigma$ be an ordering of the agents, $p_1,\dots,p_\nPl$ be prices, and $\mathbf{x}$ be an allocation that maximizes the revenue. Recall that $x_a$ denotes the item received by agent $a\in\setPl$ if exists, otherwise $x_a$ is the empty set. We show that there exist static prices and an envy-free allocation resulting in the same revenue. As the reverse always holds, that is, any envy-free allocation with respect a static pricing can be seen as a strongly envy-free solution in the dynamic setting, this proves the theorem.

We define a static pricing as follows: for each agent $a\in\setPl$ with $x_a\neq\emptyset$, set $p(x_a)\coloneqq p_{\sigma(a)}(x_a)$, that is, we define the price of an allocated item to be the price at which it was sold. For the remaining items, set the price to $+\infty$. We claim that the allocation $\mathbf{x}$ is envy-free with respect to the static pricing $p$, and hence has the same revenue as the dynamic solution. Indeed, this follows from the definition of strong envy-freeness, as $x_a\in\argmax\{v_{a}(i)-p_t(i)\mid t\in[\nPl],i\in\setIt\}$ and $p(x_a)=p_{\sigma(a)}(x_a)$ imply $x_a\in \argmax\{v_a(i)-p(i)\mid i\in\setIt\}$ for each $a\in\setPl$. 
\end{proof}

Unfortunately, the problem remains hard for weaker notions of envy-freeness. Our proof follows the main idea of the proof of Guruswami et al.~\cite{guruswami2005profit} for the APX-hardness of revenue maximization. 

\begin{theorem}\label{thm:ex_revenue_hard}
Maximizing the revenue in the unit-demand ex-post and ex-ante envy-free dynamic pricing problems are APX-hard, even if the agents' ordering is known in advance and ties are broken by the seller.
\end{theorem}
\begin{proof}
The proof is by reduction from \textsc{Vertex Cover} in 3-regular graphs. Given a $3$-regular graph $G=(V,E)$, \textsc{Vertex Cover} asks for a minimum number of vertices that includes at least one endpoint of every edge of the graph. This problem was shown to be APX-hard in~\cite{fleischner2010maximum}. 

Let $G=(V,E)$ be a $3$-regular graph with $n$ vertices and $m$ edges. We construct a pricing instance with items set $\setIt$ and agent set $\setPl$ consisting of $4n$ items and $m+n$ agents, respectively. For each vertex $z\in V$, we add four vertex-items $z_1,z_2,z_3,z_4$ to $\setIt$ and one vertex-agent to $\setPl$ that, by abuse of notation, we also denote by $z$. The valuation of the agent is then defined as $v_z(z_i)=2$ for $i\in[4]$ and $0$ for any other item. Furthermore, for each edge $e =zw \in E$, we add an edge-agent $e$ to $\setPl$ with valuation $v_e(z_i)=v_e(w_i) = 1$ for $i\in[4]$ and $0$ for any other item; for an example, see Figure~\ref{fig:vertex}. 

We first consider the ex-post case. Assume that the ordering of the agents is such that edge-agents arrive first, followed by vertex-agents. We claim that for such an ordering, there exists an ex-post envy-free pricing scheme and allocation that results in a total revenue of $m+2n-k$ if and only if there exists a vertex cover of size $k$ in $G$. Since $m=3n/2$ and the minimum vertex cover has size at least $m/3=n/2$, a constant factor gap in the size of a vertex cover translates into a constant factor gap in the optimal profit for the pricing instance, which yields the desired APX-hardness result. 

To see the `if' direction, let $C\subseteq V$ be a vertex cover of $G$ of size $k$. For each vertex $z\in V$, let $p_t(z_i)\coloneqq 1$ if $z\in C$ and $p_t(z_i)\coloneqq 2$ otherwise for $i\in[4]$ and $t\in[\nPl]$ -- note that the prices do not change over time, implying that the final solution is ex-post envy-free. According to the ordering, edge-agents arrive first, and each of them takes one of the vertex-items that correspond to one of its endpoints that lies in $C$ for a price of $1$. Then vertex-agents arrive, and take a copy of the vertex-items corresponding to them for a price of $2$. Note that, since the graph is $3$-regular and four vertex-items were added for each vertex, each agent receives an item, and hence the total revenue is $m+2n-k$. 

To see the `only if' direction, consider dynamic prices $p_1,\dots,p_\nPl$ and an ex-post envy-free allocation $\mathbf{x}$ that maximizes the revenue for the ordering considered. It is not difficult to check that the pricing vectors can be assumed to take values $1$ and $2$ only. Let $e=zw$ be the first edge-agent, if exists, who does not get any item. That is, all the remaining vertex-items from $z_1,\dots,z_4,w_1,\dots,w_4$ are priced at $2$ upon the arrival of $e$. If we reduce the price of one of these items, say $z_i$, then we have to do the same modification for all the remaining vertex-items corresponding to $z$ and for all the remaining time steps to ensure ex-post envy-freeness from the point of view of vertex-agent $z$. This way, we lose a revenue of $1$ coming from vertex-agent $z$, but we gain this back by making a profit of $1$ from edge-agent $e$. By this observation, we may assume that for each vertex $z\in V$ and time step $t\in[\nPl+\nIt]$, either $p_t(z_i)=1$ for $i\in[4]$ or $p_t(z_i)=2$ for $i\in[4]$, and that vertices belonging to the former class form a vertex cover of $G$. This implies that the revenue is at most $m+2n-k$.

If the ordering of the agents is such that vertex-agents arrive first, followed by edge-agents, then a similar argument shows the hardness of the ex-ante case.
\end{proof}

\begin{figure}[t!]
\centering
  \includegraphics[width=.8\linewidth]{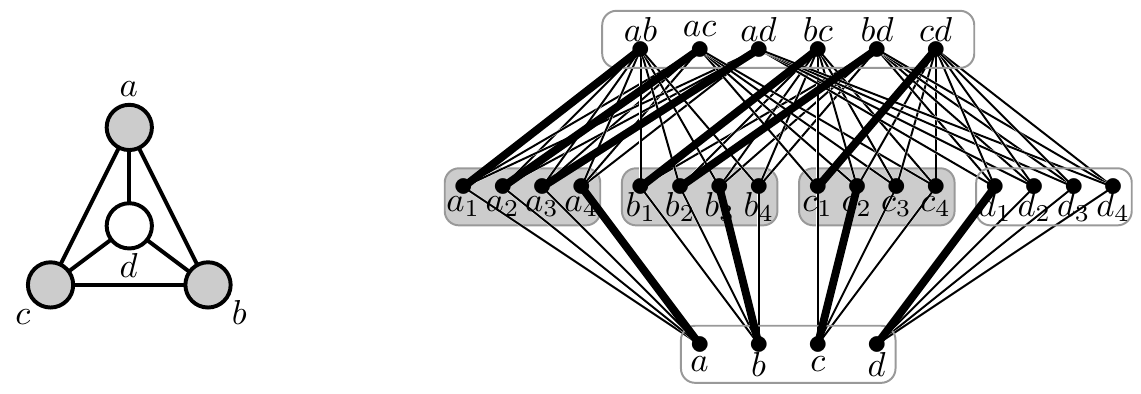}
  \caption{A $3$-regular graph $G$ with $n=4$ and $m=6$, where grey vertices form a minimum vertex-cover of size $k=3$. In the corresponding pricing instance, edges incident to vertex-agents and edge-agents have weights $2$ and $1$, respectively. When edge-agents arrive first, then setting the prices to $1$ on grey elements, $2$ otherwise, and allocating the items according to thick edges results in an ex-post envy-free solution with revenue $11=m+2n-k$.}
\label{fig:vertex}
\end{figure}

\subsection{Algorithms}
\label{sec:algorithms}

In the previous section, we showed that if the seller has no control over the order in which agents arrive, then even the seemingly more flexible framework of dynamic pricing is not enough for maximizing the revenue in the ex-post and ex-ante envy-free settings. On the positive side, if the ordering can be chosen by the seller, then an optimal pricing scheme can be determined efficiently. In what follows, we give polynomial-time algorithms for determining an ordering of the agents together with the price vectors so that the final allocation is ex-post or ex-ante envy-free and maximizes the revenue irrespective of how ties are broken by the agents. In both cases, we compare the solution to the maximum weight of a matching in the corresponding edge-weighted bipartite graph, which is clearly an upper bound for the revenue. 

Note that, in the agent-chooses case, an agent can decide not to take an item with utility $0$ for her. For keeping the description of the algorithms simple, we assume that in such cases the agent decides to take the item. In the general case then one an decrease the prices of the items in each step by a small $\varepsilon>0$, thus obtaining a revenue arbitrarily close to the maximum weight of a matching.

First we consider the ex-post case.

\begin{theorem}\label{thm:ex-post_revenue_p}
If the ordering of the agents can be chosen by the seller, then there exists a revenue-maximizing dynamic pricing scheme for the unit-demand ex-post envy-free pricing problem even if ties are broken by the agents. Furthermore, the optimal ordering and prices can be determined in polynomial time.
\end{theorem}
\begin{proof}
Consider the edge-weighted bipartite graph $G=(\setIt,\setPl;E)$ representing the market, where the weight of an edge $ai$ is $v_a(i)$ for $i\in\setIt$, $a\in\setPl$. By Lemma~\ref{lem:all}, we may assume that all the items are covered by every maximum weight matching of $G$. Let $M\subseteq E$ be an arbitrary maximum weight matching, and for each agent $a$, let $x_a$ denote the item to which $a$ is matched in $M$ if such an item exists, otherwise define $x_a$ to be the empty set. Furthermore, take a weighted covering $\pi$ provided by Lemma~\ref{lemma:dual2}. Note that $M$ consists of tight edges. 

We define the ordering $\sigma$ of the agents as follows: agents not covered by $M$ arrive first, and then the remaining agents arrive in a decreasing order according to their $\pi$ value, where ties are broken arbitrarily. 

Now we describe how to set the prices at each time step. Define the prices to be $+\infty$ until all the agents not covered by $M$ have left. Then, at each time step, consider the next agent $a$ and set the price of all remaining items to $+\infty$ except for $x_a$, for which we set the price to $v_a(x_a)$. Clearly, the agent will take the item $x_a$, hence the resulting allocation corresponds to $M$ and has total revenue equal to the maximum weight of a matching. 

It remains to verify that the pricing and the allocation thus obtained provide an ex-post envy-free solution. To see this, consider the arrival of an agent $a\in\setPl$. As all the remaining items has been priced at $+\infty$ so far except for $x_a$ which is priced at $v_a(x_a)$, it is enough to show that $a$ does not envy an item that was taken before her arrival. Those items were also priced at $+\infty$ except for the time step when they were taken by the corresponding agent. So let $a'$ be an agent who arrived before $a$ and took the item $x_{a'}$, that is, $x_{a'}\neq\emptyset$. Since $\pi$ is a weighted covering, $M$ consists of tight edges, and $\pi(a')\geq \pi(a)$, we get
\begin{align*}
u_{a,\sigma(a')}(x_{a'})
{}&{}=
v_{a}(x_{a'})-p_{\sigma(a')}(x_{a'})\\
{}&{}=
v_{a}(x_{a'})-v_{a'}(x_{a'})\\
{}&{}\leq
(\pi(a)+\pi(x_{a'}))-(\pi(a')+\pi(x_{a'}))\\
{}&{}=
\pi(a)-\pi(a')\\
{}&{}\leq 
0\\
{}&{}=
v_a(x_a)-v_a(x_a)\\
{}&{}=
u_{a,\sigma(a)}(x_a),
\end{align*}
which means that agent $a$ does not envy the item $x_{a'}$.
\end{proof}

A similar proof works for the ex-ante setting as well. However, the proof is is slightly more complicated as maintaining ex-ante envy-freeness requires a careful choice of prices. As a result, the revenue of the final allocation is not exactly the maximum weight of a matching in the associated bipartite graph, but can be arbitrarily close to that. For simplicity, we still refer to such a pricing as `optimal' in the statement of the theorem.

\begin{theorem}\label{thm:ex-ante_revenue_p}
If the ordering of the agents can be chosen by the seller, then there exists a revenue-maximizing dynamic pricing scheme for the unit-demand ex-ante envy-free pricing problem even if ties are broken by the agents. Furthermore, the optimal ordering and prices can be determined in polynomial time.    
\end{theorem}
\begin{proof}
Consider the edge-weighted bipartite graph $G=(\setIt,\setPl;E)$, a maximum weight matching $M$, $x_a$ for $a\in\setPl$, and weighted covering $\pi$ as in the proof of Theorem~\ref{thm:ex-post_revenue_p}. Let $0<\delta<\min\bigl\{\min\{\pi(a)+\pi(i)-v_a(i)\mid ia\in E\ \text{is not tight}\},\ \min\{\pi(i)\mid i\in \setIt\}\bigr\}$. Note that, by Lemma~\ref{lemma:dual2} and the assumption that each item is covered by every maximum weight matching of $G$, such a $\delta$ exists. Furthermore, let $0<\varepsilon< \delta/2^\nPl$.

We define the ordering $\sigma$ of the agents as follows: agents covered by $M$ arrive first in an increasing order according to their $\pi$ values where ties are broken arbitrarily, followed by the remaining agents. 

Now we describe how to set the prices at each time step. If an agent $a$ arrives for which $x_a\neq\emptyset$, then for each item $i\in\setIt$ set its price to $\pi(a)+\pi(i)-\delta/2^{\sigma(a)}+\varepsilon$ except for $x_a$, for which we set the price to $\pi(a)+\pi(x_a)-\delta/2^{\sigma(a)}$. If $x_a=\emptyset$, then define the prices to be $+\infty$. By the definition of the ordering and the values $\delta$ and $\varepsilon$, the prices remain non-negative and do not decrease over time, hence the resulting allocation is automatically ex-ante envy-free. 

It suffices to show that each agent $a$ chooses $x_a$ upon arrival. Indeed, if this holds, then that results in a profit of  $\pi(a)+\pi(x_a)-\delta/2^{\sigma(a)}\geq v_a(x_a)-\delta$, where we used the fact that the edge $x_aa$ is tight by Lemma~\ref{lemma:dual2}(a). By choosing $\delta$ small enough, the total revenue of the final allocation can be arbitrarily close to the weight of $M$. Consider any remaining item $i$ distinct from $x_a$. As $\pi$ is a weighted covering and $x_aa$ is tight, we get 
\begin{align*}
u_{a,\sigma(a)}(i)
{}&{}=
v_{a}(i)-p_{\sigma(a)}(i)\\
{}&{}=
v_{a}(i)-(\pi(a)+\pi(i)-\delta/2^{\sigma(a)}+\varepsilon)\\
{}&{}< 
\delta/2^{\sigma(a)}\\
{}&{}=
v_a(x_a)-(\pi(a)+\pi(x_a)-\delta/2^{\sigma(a)})\\
{}&{}=
u_{a,\sigma(a)}(x_a).
\end{align*}
This means that $x_a$ is the unique maximizer of the utility of $a$ at time step $\sigma(a)$ and has positive utility for $a$, hence agent $a$ takes $x_a$ as stated.
\end{proof}

Finally, we settle the existence of weakly envy-free solutions when the ordering of the agents is fixed but known in advance. 

\begin{theorem}\label{thm:weak_revenue}
If the ordering of the agents is known in advance, then there exists a revenue-maximizing dynamic pricing scheme for the unit-demand weakly envy-free pricing problem even if ties are broken by the agents. Furthermore, the optimal prices can be determined in polynomial time.     
\end{theorem}
\begin{proof}
Let $\sigma$ denote the fixed ordering of the agents. Define an edge-weighted bipartite graph $G=(\setIt,\setPl;E)$, maximum weight matching $M\subseteq E$, and $x_a$ for $a\in\setPl$ as in the proof of Theorem~\ref{thm:ex-post_revenue_p}. At the arrival of agent $a$, set the price of all remaining items to $+\infty$ except for $x_a$, for which we set the price to $v_a(x_a)$. The agent clearly takes $x_a$ at the maximum possible price, hence the resulting allocation and pricing are optimal.
\end{proof}

\section{Conclusions}
\label{sec:conclusions}

In this paper, we studied the existence of optimal dynamic prices under fairness constraints in unit-demand markets. We proposed four possible notions of envy-freeness depending on the time period over which agents compare themselves to others, and settled the existence of optimal dynamic prices in various settings. 

We close the paper with mentioning a few open problems. While we concentrated on social welfare and revenue maximization problems, a natural question is to consider alternative objective functions such as the average or the max-min social welfare and revenue. Besides being interesting on their own, such functions may be used to overcome the difficulties mentioned in Remark~\ref{rem:strong_welfare}. 

A recently line of research investigated the problem of balancing fairness and efficiency in markets, see e.g.~\cite{garg2021approximating}. It would be interesting to see how dynamic envy-free pricing behaves under such objective functions. 


Finally, the complexity of weak envy-free revenue maximization with unspecified order remains open. This variant is of special interest, since it naturally connects revenue maximization with the recent popular strain of research on dynamic pricing schemes.

\paragraph{Acknowledgement.}

The work was supported by DAAD with funds of the Bundesministerium f{\"u}r Bildung und Forschung (BMBF), the Lend\"ulet Programme of the Hungarian Academy of Sciences -- grant number LP2021-1/2021 and by the Hungarian National Research, Development and Innovation Office -- NKFIH, grant number FK128673.

\bibliographystyle{abbrv}
\bibliography{pricing.bib}

\end{document}